\documentclass{article}
\usepackage[utf8]{inputenc}
\usepackage{amsmath,amsthm,amssymb,graphicx,tikz,caption,subcaption,xcolor,nicefrac,hyperref,cleveref}
\usepackage{fullpage}
\usepackage{bm}
\usepackage{xspace}
\usepackage{framed}

\allowdisplaybreaks

\newtheorem{theorem}{Theorem}

\newtheorem{lemma}{Lemma}

\newcommand{\cN}{\mathcal{N}}

\newcommand{\tO}{\tilde{O}}

\newcommand{\E}{\mathbb{E}}
\newcommand{\eps}{\epsilon}

\newcommand{\tOmega}{\tilde{\Omega}}

\newcommand{\Paren}[1]{\left(#1\right)}

\title{Improved Approximation Algorithm \\for Maximum Balanced Biclique}
\date{\today}

\author{
Pasin Manurangsi\thanks{Email: \texttt{pasin@google.com}.} \\
Google Research
}

\begin{document}

\maketitle

\begin{abstract}
We study the Maximum Balanced Biclique (MBB) problem: Given a bipartite graph $G$ with $n$ vertices on each side, find a balanced biclique in $G$ with maximum size. We give a polynomial-time $\left(\frac{n}{\widetilde{\Omega}\left((\log n)^3\right)}\right)$-approximation algorithm for the problem, which improves upon an $\left(\frac{n}{\Omega\left((\log n)^2\right)}\right)$-approximation by Chalermsook et al.~\cite{ChalermsookJO20} and answers their open question. Furthermore, our approximation ratio matches that of the maximum clique problem by Feige~\cite{Feige04} up to an $O(\log \log n)$ factor.
\end{abstract}

\section{Introduction}

The \emph{Maximum Clique} problem---where we are given an undirected graph and the goal is to find a maximum clique---is one of the most well studied problems in all of combinatorial optimization. One of the original 21 NP-hard problems shown by Karp~\cite{Karp72}, Maximum Clique also played a central role in the foundation for connections between the PCP Theorem and hardness of approximation~\cite{FeigeGLSS96,AroraLMSS98,AroraS98,BellareGS98,Hastad96,Khot01,EngebretsenH03,KhotP06,Zuckerman07}. Specifically, it is NP-hard to approximate Maximum Clique to within a factor of $n^{1 - \eps}$ for any constant $\eps > 0$~\cite{Hastad96,Zuckerman07}. Under a stronger assumption that NP $\nsubseteq$ BPTIME$\Paren{2^{(\log n)^{O(1)}}}$, Khot and Ponnuswami~\cite{KhotP06} ruled out any polynomial-time algorithm with approximation ratio $n / 2^{(\log n)^{3/4 + \gamma}}$ for any constant $\gamma > 0$. This remains the best known hardness of approximation for Maximum Clique to date. Meanwhile, many approximation algorithms for Maximum Clique have been devised (e.g.~\cite{BoppanaH90,AlonK98,Feige04}). In particular, Feige's algorithm~\cite{Feige04} achieves the best known approximation ratio of $O\Paren{n \cdot \frac{(\log \log n)^2}{(\log n)^3}}$ in polynomial time.

In this work, we study a bipartite variant of the Maximum Clique problem, called the \emph{Maximum Balanced Biclique (MBB)} problem. Recall that a $k$-biclique, denoted by $K_{k, k}$, is the complete bipartite subgraph where each side contains $k$ vertices. In MBB, we are given a bipartite graph $G = (U, V, E)$ where $|U| = |V| = n$, and the goal is to find a maximum $k$-biclique in $G$. 

The exact version of MBB was stated\footnote{Originally, the book did not provide a proof, but such a proof was later published e.g. in~\cite{Joh87}.} to be NP-hard in Garey and Johnson's seminal book~\cite[page 196]{GJ79}. However, despite a considerable amount of effort over the years (e.g.~\cite{Fei02,FK04,Kho06,Man16,BhangaleGHKK17,Manurangsi18}), it is \emph{not} even known to be NP-hard to approximate to within a factor of $1 + \eps$ for any constant $\eps > 0$. Nevertheless, hardness of approximation results are known under other assumptions~\cite{Fei02,FK04,Kho06,Man16,BhangaleGHKK17,Manurangsi18}. Specifically, it is known to be NP-hard (under randomized reductions) to approximate to within a factor of $n^{1 - \eps}$ for any $\eps > 0$, assuming strengthenings of the Unique Games Conjecture~\cite{BhangaleGHKK17,Manurangsi18}.

Despite the apparent challenges in proving hardness of approximation results for MBB, not much is known in terms of approximation algorithms either. It is straightforward to see that an ``enumeration style'' algorithm yields $\Paren{\frac{n}{\tOmega(\log n)}}$-approximation for MBB in polynomial time. Only recently did Chalermsook et al.~\cite{ChalermsookJO20} manage to improve upon this approximation guarantee, obtaining an $\Paren{\frac{n}{\Omega\Paren{(\log n)^2}}}$-approximation. As an open question, they asked whether it is possible to improve this guarantee to match the $\Paren{\frac{n}{\tOmega\Paren{(\log n)^3}}}$-approximation by~\cite{Feige04} for the Maximum Clique problem.

\subsection{Our Contribution}
Our main result is an algorithm that achieves such an approximation guarantee:

\begin{theorem}[Main Theorem] \label{thm:main}
There is a randomized polynomial-time $O\Paren{n \cdot \Paren{\frac{\log \log n}{\log n}}^3}$-approximation for Maximum Balanced Biclique.
\end{theorem}

Our approximation ratio matches the aforementioned algorithm of Feige for the Maximum Clique~\cite{Feige04} problem up to an $O(\log \log n)$ factor.

\paragraph{Proof Overview.}
To explain the proof techniques, let us use $k$ to denote the size of the optimal balanced biclique in $G$, and let $t = n/k$. We first review the algorithmic guarantee of \cite{ChalermsookJO20}: If $t \leq (\log n)^{O(1)}$, then they show how to use graph removal technique from \cite{Feige04} on \emph{matchings} in order to find $\tOmega((\log n)^2)$-biclique, as formalized below.

\begin{theorem}[\cite{ChalermsookJO20}] \label{thm:cjo}
Given a graph $G$ that contains a $k$-biclique, there is a polynomial-time algorithm that finds a $\Theta\Paren{\Paren{\frac{\log n}{\log t}}^2}$-biclique in $G$ where $t = n/k$.
\end{theorem}

It is not hard to see that the above theorem immediately yields $O(n/(\log n)^2)$-approximation. For our purpose, observe further that, if $t \geq \tOmega(\log n)$, then the above already yields $\left(\frac{n}{\widetilde{\Omega}\left((\log n)^3\right)}\right)$-approximation. Thus, we can focus on the case $t \leq \tilde{O}(\log n)$, i.e. when the optimum $k$ is at least $\frac{n}{\tO(\log n)}$. 

For the case of Maximum Clique, this case can be solved via semi-definite program (SDP) rounding~\cite{AlonK98,Halperin02} or via subgraph removal~\cite{BoppanaH90}. In particular, for $t \leq O\Paren{\frac{\log n}{\log \log n}}$, these algorithms can find a clique of size $n^{1/O(t)}$. Unfortunately, Chalermsook et al.'s~\cite{ChalermsookJO20} technique of reducing MBB to Maximum Clique has a square blow up on the parameter $t$. This means that one can only apply these Maximum Clique algorithms in a black-box manner\footnote{See~\cite{Manurangsi2022-cstheory-stackexchange} for a more detailed description on how to interpret \cite{ChalermsookJO20} as a reduction from MBB to Maximum Clique.} when $t \leq \tO(\sqrt{\log n})$, which is not sufficient for our purpose.

Instead, our main result is to show that, in the regime $t \leq O\Paren{\frac{\log n}{\log \log n}}$, a more direct SDP relaxation and rounding can yield a balanced biclique of size $n^{1/O(t)}$, as stated below. As mentioned earlier, combining this with Chalermsook et al.'s result (\Cref{thm:cjo}) immediately yields the final approximation guarantee.

\begin{theorem}[Main SDP Rounding] \label{thm:sdp-main}
There exists a constant $D > 1$ such that the following holds. If $\{u_i\}, e$ is a feasible solution to the SDP relaxation in \Cref{fig:sdp-relaxation} with $k \geq n \cdot \frac{D \cdot \log \log n}{\log n}$, there is an efficient rounding algorithm for finding a balanced biclique of at least size $n^{\frac{1}{D \cdot t}}$ with high probability, where $t = n/k$.
\end{theorem}

\paragraph{Integrality Gap for A Natural Relaxation.} 
To illustrate a challenge faced when using SDP relaxation for MBB, let us first state the following ``natural'' extension of Lovasz's $\vartheta$ function~\cite{Lovasz79}, which is used\footnote{Note that the formulation below looks different from one defined in \cite{Lovasz79} but the latter can be written in this form~\cite{KleinbergG98}.} for Maximum Clique~\cite{AlonK98,Halperin02}. Here, we associate a vector $u_i \in \mathbb{R}^d$ with each $i \in U \cup V$, and let $e \in \mathbb{R}^d$ be an arbitrary unit vector.
\begin{figure}[htbp]
\begin{framed}
\begin{align}
\text{Maximize} \qquad &  k \nonumber \\
\text{Subject to} \qquad    & \|e\|^2 = 1 \nonumber \\
    & \|u_i\|^2 = \langle u_i, e \rangle \quad &\forall i \in U \cup V \nonumber \\
    & \sum_{i \in U} \langle u_i, e \rangle = \sum_{j \in V} \langle u_j, e \rangle = k \label{eq:vector_sum-two-sided} \\
    & \langle u_i, u_j \rangle = 0 &\forall (i, j) \in (U \times V) \setminus E \nonumber \\
    & \langle u_i, u_j \rangle \geq 0 &\forall (i, j) \in U \times V \nonumber
\end{align}
\end{framed}
\caption{SDP Relaxation for MBB (First Attempt)} \label{fig:sdp-relaxation-theta}
\end{figure}

If $G$ contains $K_{k,k}$, there is a feasible solution (setting vectors in the biclique to $e$ and all others to $0$) of value $k$.  We remark that the only difference compared to the Maximum Clique formulation here is the constraint \eqref{eq:vector_sum-two-sided}, which aims to ensure that the number of vertices picked in $U$ is the same as that in $V$ (so that the biclique is balanced). 

Unfortunately, this SDP relaxation is completely unhelpful. We can always find a solution with $k = n/2$: Let $f$ be any unit vector orthogonal to $e$, and let $u_i = \frac{e + f}{2}$ for all $i \in U$ and $u_j = \frac{e - f}{2}$ for all $j \in V$. This ensures that $\langle u_i, u_j \rangle = 0$ for all $(i, j) \in U \times V$. In other words, the integrality gap is unbounded.

\paragraph{Our SDP Relaxation.} Given the integrality gap, we thus use a stronger SDP relaxation for MBB. We again associate a vector $u_i \in \mathbb{R}^d$ with each $i \in U \cup V$, and let $e \in \mathbb{R}^d$ be a unit vector. The SDP relaxation is stated below. Note that we formulate it as a feasibility problem rather than an optimization problem.
\begin{figure}[htbp]
\begin{framed}
\begin{align}
    & \|e\|^2 = 1 \nonumber \\
    & \|u_i\|^2 = \langle u_i, e \rangle \quad &\forall i \in U \cup V \nonumber \\
    & \sum_{i \in U} \langle u_i, e \rangle = \sum_{j \in V} \langle u_j, e \rangle = k \label{eq:vector_sum} \\
    & \langle u_i, u_j \rangle = 0 &\forall (i, j) \in (U \times V) \setminus E \label{eq:nonedge} \\
    & \langle u_i, u_j \rangle \geq 0 &\forall (i, j) \in U \times V \label{eq:nonneg} \\
    & \sum_{j \in V} \langle u_i, u_j \rangle = k \langle u_i, e\rangle &\forall i \in U \label{eq:frac_deg_u} \\
    & \sum_{i \in U} \langle u_i, u_j \rangle = k \langle u_j, e\rangle  &\forall j \in V \label{eq:frac_deg_v}
\end{align}
\end{framed}
\caption{SDP Relaxation for MBB} \label{fig:sdp-relaxation}
\end{figure}

We note that the only additional constraints are \eqref{eq:frac_deg_u} and \eqref{eq:frac_deg_v}, which are ``fractional degree'' constraints aiming to ensure that if $i$ is picked, then $k$ of its neighbor must be picked.
Again, if $G$ contains $K_{k,k}$, there is a feasible solution (setting vectors in the biclique to $e$ and all others to $0$). 

Our main technical contribution (\Cref{thm:sdp-main}) is to show that a feasible solution to this SDP can be rounded to a large biclique. It turns out that the rounding procedure is essentially the same as that of \cite{Halperin02} in the clique case: After some preprocessing, pick a random Gaussian $g$ and keep a vertex $i$ if its dot product with $g$ is at least a certain threshold. (The full procedure is presented in \Cref{fig:sdp-rounding}.)

The main difference however is in the analysis. In \cite{Halperin02}, one simply computes the expected number of vertices that survive the rounding and the expected survived number of non-edge pairs. If the former is much larger than the latter, then (with not too small probability) we can find a large clique by deleting one end of every non-edge pair. This analysis does not work for us: Even if the expected number of survived vertices is large, it is possible that the rounding either keep all vertices in $U$ and no vertex in $V$, or vice versa. Indeed, this is the situation that occurs with the above integrality gap instance of the SDP in \Cref{fig:sdp-relaxation-theta}.

To overcome this issue, we instead lower bound the expected number of survived \emph{edges}. We do this by showing that, due to the extra constraints \eqref{eq:frac_deg_u} and \eqref{eq:frac_deg_v}, many edges $(i, j) \in E$ are ``good'' in the sense that $\langle u_i, u_j \rangle$ is not too small. This implies that they are positively correlated in the rounding procedure, which allows us to lower bound the expected number of survived edges. This conclude our proof overview for the rounding procedure, which is formalized in \Cref{sec:rounding} after some preliminaries.

\section{Preliminaries}

Throughout this work, we use $\log$ to denote the natural logarithm.

\subsection{Gaussian Tail Bounds} 

Our rounding scheme relies on probability bounds for (multivariate) normal distributions, which we will state here for convenience. We note that these bounds are standard and are also used in approximation algorithms for Maximum Clique (e.g. \cite{Halperin02}). First, let us start with the following lemma:
\begin{lemma}[\cite{feller1991introduction}] \label{lem:gaussian-corr-bound}
Let $\phi(x) = \frac{1}{\sqrt{2\pi}} \cdot \exp(-x^2/2)$ be the PDF of the standard Gaussian distribution $\cN(0, 1)$. If $Z \sim \cN(0, 1)$ and $\tau \geq 2$, then
$\frac{\phi(\tau)}{2\tau} \leq \Pr[Z \geq \tau] \leq \frac{\phi(\tau)}{\tau}$.
\end{lemma}

For two standard Gaussian variables $X, Y$ with correlation $\rho = \mathbb{E}[XY]$, we use the following bivariate tail bounds for $\Pr[X \geq \tau, Y \geq \tau]$ for $\tau \geq 2$:
\begin{itemize}
    \item \textbf{Positive Correlation ($\rho \ge 0$):} Since $X, Y$ are positively correlated, the joint probability is lower-bounded by the product of their marginals:
    \begin{equation}
        \Pr[X \ge \tau, Y \ge \tau] \ge \Pr[X \ge \tau]\Pr[Y \ge \tau] \ge \frac{\phi(\tau)^2}{4\tau^2}, \label{eq:pos_corr}
    \end{equation}
    where the last inequality is due to \Cref{lem:gaussian-corr-bound}.
    \item \textbf{Negative Correlation ($\rho < 0$):} Notice that $\frac{X+Y}{\sqrt{2 + 2\rho}} \sim \mathcal{N}(0, 1)$. Thus,
    \begin{align}
        \Pr[X \ge \tau, Y \ge \tau] \le \Pr[X+Y \ge 2\tau] 
        < \frac{1}{\sqrt{2}\tau} \cdot \phi\Paren{\frac{\sqrt{2} \cdot \tau}{\sqrt{1 + \rho}}}
        \leq \frac{\sqrt{\pi}}{\tau} \cdot \phi(\tau)^2 \cdot \exp\Paren{\rho \cdot \tau^2}  \label{eq:neg_corr}
    \end{align}
    where the first inequality is due to \Cref{lem:gaussian-corr-bound} and $\rho < 0$, and the second inequality follows from $\rho < 0$.
\end{itemize}
\subsection{Biclique Extraction Lemma}
Before we describe the rounding, let us also state a self-contained combinatorial step of extracting a balanced biclique from a sufficiently dense surviving subgraph.

\begin{lemma}[Greedy Extraction] \label{lem:extraction}
Let $H = (U_H, V_H, E_H)$ be a bipartite graph with at most $n$ vertices on each side. Let $F$ be the total number of edges (i.e. $F = |E_H|$) and $Q$ be the total number of non-edges in $H$ (i.e. $Q = |U_H| \cdot |V_H| - |E_H|$). If there exists an integer $r \ge 1$ such that $F - 2rQ \ge 2nr$, then there is a polynomial-time algorithm that finds a $K_{r, r}$ in $H$.
\end{lemma}

\begin{proof}
Let $d_F(i)$ and $d_Q(i)$ be the degrees of a vertex $i$ in the edges and non-edges of $H$, respectively. We apply the following two-stage procedure:
\begin{enumerate}
    \item \textbf{Density Cleaning:} While there exists any vertex $i$ such that $d_F(i) \le 2r \cdot d_Q(i)$, delete $i$ from $H$.
    
    To analyze this procedure, we track $W = F - 2rQ$. Upon deleting $i$, the change in $W$ is $-d_F(i) + 2r \cdot d_Q(i) \ge 0$. Thus, the final cleaned graph $H' = (U_{H'}, V_{H'}, E_{H'})$ satisfies $W' \ge 2nr$. 
    
    Since $W' \leq |E_{H'}| \leq |U_{H'}| \cdot |V_{H'}|$ and we assume that each side of $H$ is of size at most $n$, we can conclude that $|U_{H'}|, |V_{H'}|$ are both at least $2r$.
    
    \item \textbf{Biclique Construction:} 
    We arbitrarily select exactly $r$ vertices from $U_{H'}$ to form $U_{cliq}$. For each selected $u \in U_{cliq}$, we delete its non-neighbors from $V_{H'}$; from how $H'$ is constructed, there are at most $\frac{|V_{H'}|}{2r}$ such non-neighbors. Thus, the total number of vertices deleted from $V_{H'}$ is at most $r \left( \frac{|V_{H'}|}{2r} \right) = \frac{|V_{H'}|}{2}$. 
    Consequently, the number of vertices surviving in $V_{H'}$ is at least $\frac{|V_{H'}|}{2}$. Because we established $|V_{H'}| \ge 2r$, at least $\frac{2r}{2} = r$ vertices remain. We arbitrarily select $r$ of these surviving vertices to form $V_{cliq}$. The set $U_{cliq} \cup V_{cliq}$ is the desired balanced biclique. \qedhere
\end{enumerate}
\end{proof}

We remark that, for the regime $\Theta(1)\leq r \leq n^{1 - \Omega(1)}$, known techniques~\cite{FeigeK10,AxenovichSSW21} can be used in the ``Biclique Construction'' step to improve the bound from $r$ to $\Omega(r \log n)$. However, this does not affect the overall approximation ratio of our algorithm and, thus, we opt for the argument above for simplicity.

\section{The Rounding Algorithm and Proof of \Cref{thm:sdp-main}}
\label{sec:rounding}

The full SDP rounding procedure is presented in \Cref{fig:sdp-rounding}.

\begin{figure}[htbp]
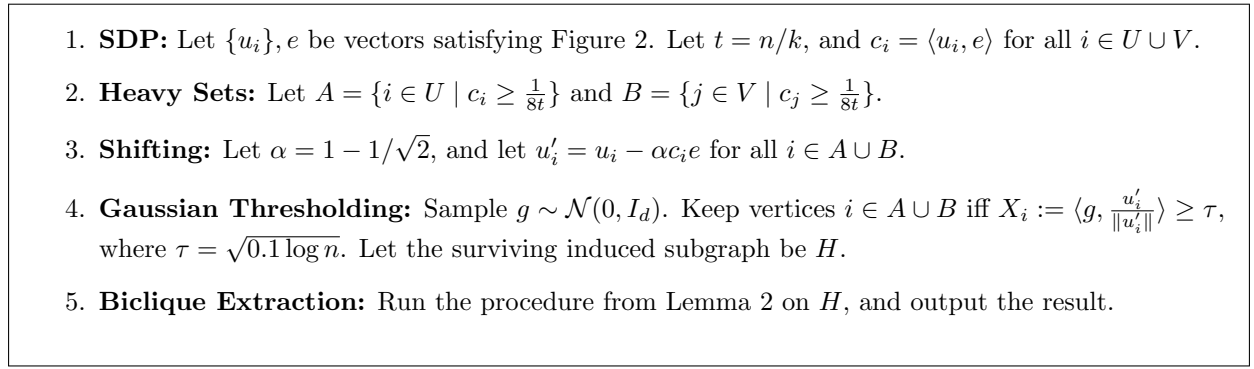

\begin{framed}
\begin{enumerate}
    \item \textbf{SDP:} Let $\{u_i\}, e$ be vectors satisfying \Cref{fig:sdp-relaxation}. Let $t = n/k$, and $c_i = \left<u_i, e\right>$ for all $i \in U \cup V$.
    \item \textbf{Heavy Sets:} Let $A = \{ i \in U \mid c_i \ge \frac{1}{8t} \}$ and $B = \{ j \in V \mid c_j \ge \frac{1}{8t} \}$.
    \item \textbf{Shifting:} Let $\alpha = 1 - 1/\sqrt{2}$, and let $u'_i = u_i - \alpha c_i e$ for all $i \in A \cup B$.
    \item \textbf{Gaussian Thresholding:} Sample $g \sim \mathcal{N}(0, I_d)$. Keep vertices $i \in A \cup B$ iff $X_i := \langle g, \frac{u'_i}{\|u'_i\|} \rangle \ge \tau$, where $\tau = \sqrt{0.1 \log n}$. 
    Let the surviving induced subgraph be $H$.
    \item \textbf{Biclique Extraction:} Run the procedure from Lemma \ref{lem:extraction} on $H$, and output the result.
\end{enumerate}
\end{framed}
\caption{SDP Rounding Algorithm for MBB} \label{fig:sdp-rounding}
\end{figure}

\begin{proof}[Proof of \Cref{thm:sdp-main}]
 We will next analyze the guarantee of the rounding algorithm. To do so, let $D = 1000$ and assume that $t \leq \frac{\log n}{D \cdot \log \log n}$.

For any subsets $S \subseteq U$ and $T \subseteq V$, let $M(S, T) := \sum_{(i,j) \in S \times T} \langle u_i, u_j \rangle$. We have
\begin{align}
M(A, B) &= M(U, V) - M(U \setminus A, V) - M(U, V \setminus B) + M(U \setminus A, V \setminus B) \nonumber \\
&\overset{\eqref{eq:nonneg}}{\geq} M(U, V) - M(U \setminus A, V) - M(U, V \setminus B) \nonumber \\
&= \sum_{i \in U} \left(\sum_{j \in V} \left<u_i, u_j\right>\right) - \sum_{i \in U \setminus A} \left(\sum_{j \in V} \left<u_i, u_j\right>\right) - \sum_{j \in V \setminus B}  \left(\sum_{i \in U} \left<u_i, u_j\right>\right) \nonumber \\
&\overset{\eqref{eq:frac_deg_u}, \eqref{eq:frac_deg_v}} = \sum_{i \in U} k \cdot c_i - \sum_{i \in U \setminus A} k \cdot c_i - \sum_{j \in V \setminus B} k \cdot c_j \nonumber \\
&\overset{\eqref{eq:vector_sum}}{=} k^2 - \sum_{i \in U \setminus A} k \cdot c_i - \sum_{j \in V \setminus B} k \cdot c_j \nonumber \\
&\overset{(\star)}{\geq} k^2 - n \cdot k \cdot \frac{1}{8t} -  n \cdot k \cdot \frac{1}{8t} \qquad = \frac{3}{4}k^2, \label{eq:mass-lb}
\end{align}
where $(\star)$ follows from our definition of $A, B$.

Let $E^+ = \{ (i,j) \in A \times B \mid \langle u_i, u_j \rangle > \frac{1}{2} c_i c_j \}$. 
Note that $\sum_{(i, j) \in A \times B} c_i c_j \le (\sum_{U} c_i)(\sum_{V} c_j) \overset{\eqref{eq:vector_sum}}{=} k^2$. Hence,
\[ \sum_{(i, j) \in A \times B} \left( \langle u_i, u_j \rangle - \frac{1}{2} c_i c_j \right) \geq M(A, B) - \frac{k^2}{2} \overset{\eqref{eq:mass-lb}}{\ge} \frac{k^2}{4} \]
Any pair $(i, j) \notin E^+$ contributes a non-positive value to this sum. This, together with the fact that $\langle u_i, u_j\rangle \leq 1$ for all $i, j \in U \cup V$, implies that $|E^+| \geq \frac{k^2}{4}$. Note also that, by \eqref{eq:nonedge}, all pairs $(i, j) \in E^+$ must be an edge.

Next, consider any $i, j \in U \cup V$. We have $\langle u'_i, u'_j\rangle = \langle u_i, u_j \rangle - \frac{1}{2} c_i c_j$. Note that this also implies $\|u'_i\|_2^2 = c_i(1 - \frac{1}{2}c_i) \leq c_i$. We next analyze the expected number of edges and non-edges in $H$. To do this, let $F$ denote the number of surviving edges in $H$, $Q$ denote the number of non-edges in $H$, and let $$r = \frac{1}{64\sqrt{\pi} \cdot \tau \cdot t^2} \cdot \exp\Paren{\frac{\tau^2}{16t}}.$$ Note that, for sufficiently large $n$, we have
\begin{align} \label{eq:biclique-size-bound}
r \geq \frac{1}{O((\log n)^3)} \cdot n^{\frac{1}{160t}} &\geq n^{\frac{1}{D \cdot t}}.
\end{align}

\begin{itemize}
    \item \textbf{Lower Bounding $\E[F]$:} For $(i,j) \in E^+$, since $\langle u_i, u_j \rangle \ge \frac{1}{2}c_i c_j$, we have $\langle u'_i, u'_j \rangle \ge 0$. Thus, $\E[X_iX_j] \geq 0$. Applying \eqref{eq:pos_corr}, the probability that $(i, j)$ is kept in $H$ is at least $\frac{\phi(\tau)^2}{4\tau^2}$. Hence, 
    \begin{align}
    \E[F] &\geq |E^+| \cdot \frac{\phi(\tau)^2}{4\tau^2} \geq \frac{k^2}{4} \cdot \frac{\phi(\tau)^2}{4\tau^2}. \label{eq:lb-f-intermediate}
    \end{align}
    Since $k = n/t$ and due to our bound on $t$, for any sufficiently large $n$, we have
    \begin{align} \label{eq:lb-f-final}
    \E[F] \geq \frac{n^2}{O((\log n)^3)} \cdot \exp(-0.1 \log n) \geq n^{1.8} \geq n \cdot \exp\Paren{\frac{\tau^2}{16t}} \geq 8nr.
    \end{align}
    \item \textbf{Upper Bounding $\E[Q]$:} Consider any $(i, j) \in (A \times B) \setminus E$. By \eqref{eq:nonedge}, $\langle u_i, u_j \rangle = 0$, which implies $\langle u'_i, u'_j \rangle = -\frac{1}{2} c_i c_j$. Normalizing bounds the correlation:
    \[ \E[X_iX_j] = \left\langle \frac{u'_i}{\|u'_i\|}, \frac{u'_j}{\|u'_j\|} \right\rangle \le \frac{-\frac{1}{2}c_i c_j}{\sqrt{c_i c_j}} = -\frac{1}{2}\sqrt{c_i c_j} \le -\frac{1}{16t}, \]
    where the last inequality follows from definition of $A, B$. Applying \eqref{eq:neg_corr}, we have
    \begin{align}
    \E[Q] \leq n^2 \cdot \frac{\sqrt{\pi}}{\tau} \cdot \phi(\tau)^2 \cdot \exp\Paren{-\frac{\tau^2}{16t}} \overset{\eqref{eq:lb-f-intermediate}}{=} \frac{1}{4r} \cdot \E[F]. \label{eq:surviving-edge-bound}
    \end{align}
\end{itemize}
Thus, we have 
\begin{align*}
\E[F - 2rQ] \overset{\eqref{eq:surviving-edge-bound}}{\geq} \E[F] / 2 \overset{\eqref{eq:lb-f-final}}{\geq} 4nr.
\end{align*}
Since $F - 2rQ$ has value at most $n^2$, we have $F - 2rQ \geq 2nr$ with probability at least $1/n^2$. When this occurs, \Cref{lem:extraction} guarantees an $r$-biclique which satisfies the desired size bound due to \eqref{eq:biclique-size-bound}. To achieve high probability, we simply repeat the rounding process $n^3$ times. This concludes our proof.
\end{proof}

\section{Putting Things Together: Proof of \Cref{thm:main}}

Finally, we prove our main theorem (\Cref{thm:main}) by combining our SDP rounding (in the large optimum case) with the algorithm of Chalermsook et al. from \Cref{thm:cjo} (in the small optimum case).

\begin{proof}[Proof of \Cref{thm:main}]
We solve the SDP from \Cref{fig:sdp-relaxation} for $k = \left\lceil n \cdot \frac{10 \cdot D \cdot \log \log n}{\log n} \right\rceil$. If the SDP is feasible, then we simply run the rounding algorithm from \Cref{thm:sdp-main}. This ensures that we find a biclique of size at least $(\log n)^{10}$, which yields $n / (\log n)^{10} \leq O\Paren{n \cdot \Paren{\frac{\log \log n}{\log n}}^3}$-approximation.

If the SDP is infeasible, then the maximum balanced biclique has size at most $O\Paren{n \cdot \frac{\log \log n}{\log n}}$. We run Chalermsook et al.'s algorithm (\Cref{thm:cjo}). If the maximum balanced biclique has size at least $n/\log^3 n$, we have $t \leq \log^3 n$ and thus the algorithm produces $\Theta\Paren{\Paren{\frac{\log n}{\log \log n}}^2}$-biclique, which yields $O\Paren{n \cdot \Paren{\frac{\log \log n}{\log n}}^3}$-approximation. On the other hand, if the maximum balanced biclique has size at most $n/\log^3 n$, we automatically get $n/\log^3 n \leq O\Paren{n \cdot \Paren{\frac{\log \log n}{\log n}}^3}$-approximation.
\end{proof}

\section{Conclusion and Open Questions}

In this work, we give a polynomial-time $\left(\frac{n}{\widetilde{\Omega}\left((\log n)^3\right)}\right)$-approximation algorithm for the Maximum Balanced Biclique problem. Our algorithm is via SDP rounding in the case where the optimum is large, combined with the known result of Chalermsook et al.~\cite{ChalermsookJO20} in the case where the optimum is small. Apart from trying to improve the approximation ratio further, it remains an interesting question whether the former case can be handled directly via combinatorial algorithms. We remark that, in the Maximum Clique case, subgraph removal techniques--which are completely combinatorial--suffice for this case~\cite{BoppanaH90}.

Another intriguing question is whether one can achieve a similar approximation for the Maximum \emph{Edge} Biclique (MEB) problem. In MEB, we are looking for a (not necessarily balanced) biclique in $G$ that contains the maximum number of edges. A trivial algorithm yields $n$-approximation, and a simple ``enumeration style'' algorithm can improve this to $\frac{n}{\tOmega(\log n)}$. However, we are not aware of any further improvement to this. In particular, it is unclear how to apply SDP rounding when the optimum is large.

\bibliographystyle{alpha}
\bibliography{ref}

\end{document}